\documentclass[conference]{IEEEtran}

\title{ \resizebox{1\hsize}{!}{Timing Analysis for DAG-based and GFP Scheduled Tasks }}

 \author{\IEEEauthorblockN{      Jos\'e Marinho \IEEEauthorrefmark{1}, Stefan M. Petters \IEEEauthorrefmark{1}  }
 \IEEEauthorblockA{\IEEEauthorrefmark{1}CISTER-ISEP Research Centre,
 Polytechnic Institute of Porto, Portugal\\
 }
 \IEEEauthorblockA{
 Email: \{jmsm,smp\}@isep.ipp.pt  
 }
 }

\usepackage[cmex10]{amsmath}

\usepackage{url}

\usepackage[pdftex]{graphicx}
\usepackage{amscd}
\usepackage{amsmath}
\usepackage{amsthm}
\usepackage{amssymb}
\usepackage{graphics}
\usepackage{graphicx}
\usepackage[vlined, ruled]{algorithm2e}
\usepackage{subfig}
\usepackage{verbatim}

\newcommand{\equals}{\stackrel{\mathrm{def}}{=}}
\newcommand{\ds}{\displaystyle}

\newcommand{\da}{\mathrm{DAG}}
\newcommand{\gftp}{\mathrm{GFP}}

\newcommand{\critp}[1]{\mathcal{P}_{#1}^{\mathrm{crit}}}
\newcommand{\critpp}[1]{\mathcal{P}_{#1}^{\mathrm{crit}(\ell, r)}}
\newcommand{\Ccritpp}[1]{C_{#1}^{\mathrm{crit}(\ell, r)}}

\newtheorem{lemma}{Lemma}

\newtheorem{theorem}{Theorem}

\newenvironment{proofsketch}{\proof[Proof Sketch] \mbox{}}{\endproof}

\newtheorem{Definition}{Definition}

\newcommand{\y} y    
\newcommand{\Y} Y    
\newcommand{\x} C    
\newcommand{\Z} Z    

\newcommand{\maxcomp}[1]{C_{#1}^{\mathrm{crit}}}
\newcommand{\maxcritk}[1]{K_{#1}^{\mathrm{crit}}}

\begin{document}

\maketitle

\begin{abstract}
Modern embedded systems have made the transition from single-core to multi-core architectures, providing performance improvement via parallelism rather than higher clock frequencies. $\da$s are considered among the most generic task models in the real-time domain and are well suited to exploit this parallelism. In this paper we provide a schedulability test using response-time analysis exploiting exploring and bounding the self interference of a $\da$ task. Additionally we bound the interference a high priority task has on lower priority ones. 

\end{abstract}

\section{\label{s:INTRO}Introduction}

The strive for higher computational power has brought about the multicore platforms as a compelling solution first in general purpose and now also in the embedded real-time systems arena. Rather than relying on the increase of the throughput of single processors, the multicore paradigm, while providing its ability to perform a greater number of simultaneous calculations, has given rise to a new challenge. It often forces system designers to utilize the hardware facilities and use parallel algorithms in order to perform tasks of high computational demand in a predefined time window. 
However, this implies a subtle difference in the way schedulability conditions are posed since parts of the workload from the same task are allowed to execute concurrently; each task is then referred to as a {\em parallel} or {\em Directed Acyclic Graph} ($\da$) task. This paper presents a framework to address this issue for fully preemptive {\em global 
fixed task priority} ($\gftp$) schedulers and homogeneous multicores in which all cores have the same computing capabilities and are interchangeable. It is worth to mention that $\gftp$ schedulers are commonly adopted and supported out of the box on several industry grade real-time operating systems such as $\operatorname{VXWorks}$~\cite{VXWorks}.

\paragraph{Related Work}
Valuable works such as~\cite{Hoon,saifula,bjorn,JingLi,vincenzo} addressed the scheduling problem of $\da$ tasks upon homogeneous multicores. Saifullah et al.~\cite{saifula} presented a method to decompose a generic $\da$ task into a set of virtual sequential tasks and after the decomposition, the popular global earliest deadline first (GEDF) density-based schedulability test is applied. Andersson and Niz~\cite{bjorn} presented an analysis for GEDF where an upper bound on the workload that each task may execute in a given time window is computed. Nevertheless, this upper-bound is computed for a special case of DAG tasks, namely the ``fork-join'' tasks. For such a task: ($i$)~the parallel workloads have the same execution requirement; ($ii$)~they are spawned after a common point; and ($iii$)~they join again after a common point. Note: When a task is executing a section of workload in parallel no further path forks can occur. Chwa et al.~\cite{Hoon} provided a method to compute the interference that each task 
would suffer in a system of so-called ``synchronous parallel'' tasks -- Each task is composed of multiple and potentially contiguous regions of parallel workloads with distinct parallelism levels --. In more than one aspect $\da$ tasks cover a broader area as they allow for parallel workloads to yield distinct execution requirements and a different immediate predecessor for each node. Previous works using $\gftp$ schedulers exist, but in a partitioned environment, i.e., tasks are assigned to cores at design time and no migration is allowed at runtime~\cite{fauberteau,Lakshmanan}. For example, Lakshmanan~et~al.~\cite{Lakshmanan} presented a basic form of $\da$ tasks, namely ``Gang tasks'', in which all the parallel workloads have to be scheduled simultaneously on the processing platform. 

\paragraph{This Research}

In this paper, we present a sufficient schedulability test applicable to constrained deadline DAG~tasks (see Section~\ref{s:MODEL} for a formal definition) scheduled by using a GFTP scheduler on a homogeneous multicore platform. 

\section{\label{s:MODEL}System Model}

\begin{figure}[h!]
\begin{center}
\includegraphics[width=0.7\linewidth]{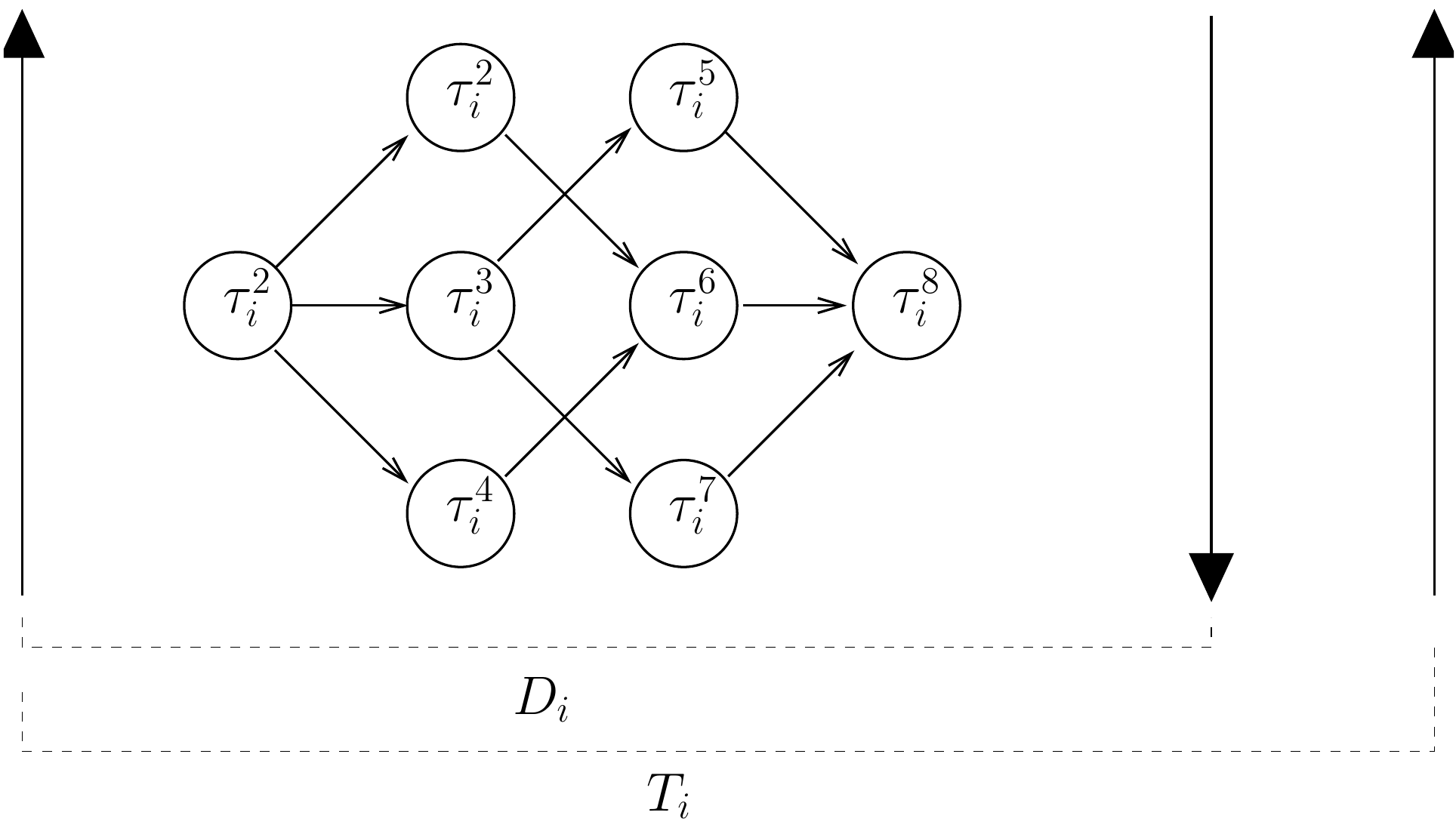}
\caption{Task $\tau_i$}
\label{fig:SM}
\end{center}
\end{figure}
\vspace{-0.8cm}
\paragraph{Task specifications} We consider a task-set $\mathcal{T} \equals \{\tau_1, \ldots , \tau_n\}$ composed of $n$ sporadic tasks. Each sporadic task $\tau_i \equals \langle G_i, D_i, T_i\rangle$, $1 \le i \le n$, is characterized by a $\da$ $G_i$, a \textit{relative deadline} $D_i$ and a \textit{minimum timespan}~$T_i$ (also called period) between two consecutive activations of $\tau_i$. These parameters are given with the following interpretation. Nodes in $G_i$ (also called sub-jobs in the literature) stand for a vector of execution requirements at each activation of $\tau_i$, and the edges represent dependencies between the nodes. A node is denoted by $\tau_i^j$, with $1 \le j \le n_i$, where $n_i$ is the total number of nodes in $G_i$. The execution requirement of node $\tau_i^j$ is denoted by $c_i^j \in [c_{i,\operatorname{min}}^j, c_{i,\operatorname{max}}^j]$. A direct edge from node $\tau_i^j$ to node $\tau_i^k$, denoted as $\tau_i^j \rightarrow \tau_i^k$, implies that the execution of $\tau_
i^k$ 
cannot start unless that of $\tau_i^j$ has completed. In this case, $\tau_i^j$ is called a \textit{parent} of $\tau_i^k$, while $\tau_i^k$ is its child. We denote the \textit{set of all children} of node $\tau_i^j$ by $succ(\tau_i^j) \equals \{ \tau_i^k \mid \tau_i^j \rightarrow \tau_i^k \}$ and the \textit{set of all parents} of node $\tau_i^j$ by $pred(\tau_i^j) \equals \{ \tau_i^k \mid \tau_i^k \rightarrow \tau_i^j \}$. If $(\tau_i^k \not\in succ(\tau_i^j) \cup pred(\tau_i^j)) \wedge (\tau_i^j \not\in succ(\tau_i^k) \cup pred(\tau_i^k))$, then $\tau_i^j$ and $\tau_i^k$ may execute concurrently. In this case, we state that $\tau_i^k \in conc(\tau_i^j)$, and reversely, $\tau_i^j \in conc(\tau_i^k)$. A node without parent is called an \textit{entry node}, while a node without child is called an \textit{exit node}. We assume that a node can start executing only after all its parents have completed. For brevity sake, we consider \textit{only} $\da$ tasks with a single entry and exit nodes. For each task $\tau_
i$, we assume $D_i \le T_i$,
 which is commonly referred to as the constrained deadline task model. Figure~\ref{fig:SM} illustrates a $\da$ task $\tau_i$ with $n_i = 8$ nodes. Note: the analysis presented in this paper is easily tunable for $\da$ tasks with multiple entry and exit nodes. 

The \textit{total execution requirement} of $\tau_i$, denoted by $C_i$, is the \textit{sum} of the execution requirements of all the nodes in $G_i$, i.e., $C_i \equals \sum_{j=1}^{n_i}c_i^j$. The task set $\mathcal{T}$ is said to be $\mathcal{A}$-\textit{schedulable}, if $\mathcal{A}$ can schedule $\mathcal{T}$ such that all the nodes of every task $\tau_i \in \mathcal{T}$ meet its deadline~$D_i$. 

\begin{Definition}[Critical path] A \textit{critical path} for task $\tau_i$, denoted by $\critp{i}$, is a directed path that has the maximum execution requirement among all paths in $G_i$.
\label{def:critpath}
\end{Definition}

\begin{Definition}[Critical path length] The \textit{critical path length} for task $\tau_i$, denoted by $\maxcomp{i}$, is the sum of execution requirements of the nodes belonging to a critical path in $G_i$.
\end{Definition}

\paragraph{Platform and scheduler specifications} We consider a platform $\pi \equals [\pi_1, \pi_2, \ldots, \pi_m]$ consisting of $m$-unit capacity cores, and a \textit{fully preemptive $\gftp$} scheduler. That is: ($i$)~a priority is assigned to each $\da$ task at system design-time and then, at run-time, every node inherits the priority of the $\da$ task it belongs to; ($ii$)~different nodes of the same $\da$ task may execute upon different cores; and finally ($iii$)~a preempted node may resume execution upon the same or a different core, at no cost or penalty. We assume that each node may execute on at most one core at any time instant and that the lower the index of a task the higher its priority.

\section{Timing Analysis and Self-Interference Extraction}

Intrinsically, some nodes of a given $\da$ task $\tau_i$ may prevent some others of the same task from executing. This constitutes a form of self-interference. Since $G_i$ may be viewed as a set of paths, say $\mathcal{P}_i$, each path $\mathcal{P}_i^k \in \mathcal{P}_i$ represents a set of sequential nodes in $G_i$ connected to each other via an edge, i.e., from the view-point of any node of $\mathcal{P}_i^k$, the other nodes of $\mathcal{P}_i^k$ are either children or parents. We denote the complementary set of $\mathcal{P}_i^k$ which contains all the nodes that do not belong to $\mathcal{P}_i^k$ by $\overline{\mathcal{P}_i^k}$. Note: the nodes in $\overline{\mathcal{P}_i^k}$ are not necessarily concurrent to all the nodes in $\mathcal{P}_i^k$. 

Let $\mathcal{P}(\tau_i^\ell,\tau_i^r)$ be the set of all partial paths in $G_i$ which connect nodes $\tau_i^\ell$ and $\tau_i^r$, and let $\mathcal{P}_i^k(\tau_i^\ell,\tau_i^r) \in \mathcal{P}(\tau_i^\ell,\tau_i^r)$ be a specific path. For brevity sake, we denote $\mathcal{P}_i^k(\tau_i^\ell,\tau_i^r)$ by $\mathcal{P}_i^{k(\ell,r)}$ for the remainder of this paper. Since $\tau_i^r \in succ(\tau_i^\ell)$ by definition of $succ(\cdot)$, each path $\mathcal{P}_i^{k(\ell,r)}$ has a worst-case execution requirement $C_i^{k(\ell,r)}$ which is computed by summing up the execution requirements of all its nodes, i.e., $C_i^{k(\ell,r)} \equals \sum_{\tau_i^j \in \mathcal{P}_i^{k(\ell,r)}} c_i^j$. Note: $\mathcal{P}(\tau_i^\ell,\tau_i^r)$ also has a critical path defined as in Definition~\ref{def:critpath}, i.e., the path with the largest execution requirement between $\tau_i^\ell$ and $\tau_i^r$. It is fairly straightforward that if either $\tau_i^\ell$ or $\tau_i^r$ is not part of the end-to-end critical path $\
critp{i}$ then it follows that the 
critical partial path between $\tau_i^\ell$ and $\tau_i^r$ is not contained in $\critp{i}$ either. Now we can quantify the maximum self-interference that a $\da$ task may generate on a given subset of $G_i$.

Let $R_i$ denote the worst-case response time of the $\da$ task $\tau_i$ -- The response time of every activation of $\tau_i$ is the timespan between its workload completion and its release -- Hence $R_i$ is the largest value from all the activations of $\tau_i$. On the roadway for the computation of an upper-bound on $R_i, \; \forall i \in [1,n]$, there are some important checkpoints we must investigate. 

\begin{Definition}[partial worst-case response time] The partial worst-case response time of the set of partial paths $\mathcal{P}(\tau_i^\ell,\tau_i^r)$ is the largest timespan between node $\tau_i^r$ completion time and node $\tau_i^\ell$ release time.
\end{Definition}

\begin{lemma}[Critical  Self-interference Path]
\label{lemma:1}
Considering only self-interference, the partial path of $\mathcal{P}(\tau_i^\ell,\tau_i^r)$ which leads to the worst-case response time of $\tau_i$ is the critical partial path $\critpp{i}$ among all partial paths in $\mathcal{P}(\tau_i^\ell,\tau_i^r)$. 
\label{lemma:self-interference1}
\end{lemma}
\begin{proof}[Proof (made by contradiction)]
Initially $\Ccritpp{i} \geqslant C_i^{d(\ell, r)}$ for any other partial path $\mathcal{P}_i^{d(\ell,r)}$. Baker and Cirinei~\cite{baker} provided an upper-bound on the interference of a {\em Liu \& Layland (LL) task} (in the LL model, each task $\tau_i$ generates a potentially infinite sequence of jobs and is characterized by a 3-tuple $\tau_i=\left\langle C_i,D_i,T_i\right\rangle$, where $C_i$ is the worst-case execution time of each job, $D_i$ is the relative deadline and $T_i \geq D_i$ is the \textit{minimum} inter-arrival time between two consecutive jobs of~$\tau_i$) on a $m$-multicore platform ($m>1$). In this work, we extend this result to compute the interference that concurrent nodes induce on $\critpp{i}$ in the same manner (see~Eq.~\ref{eq:response}). 
\begin{align}
\resizebox{0.4\hsize}{!}{$\Ccritpp{i} + \frac{1}{m} \ds\sum_{\tau_i^j \in \overline{\critpp{i}}} c_i^j$}
\label{eq:response}
\end{align}
\vspace{-.1cm}
\noindent Let us assume that for some $\mathcal{P}_i^{d(\ell,r)} \neq \critpp{i}$ we have:
\begin{align}
& \resizebox{0.72\hsize}{!}{$\Ccritpp{i} + \frac{1}{m} \ds\sum_{\tau_i^j \in \overline{\critpp{i}}} c_i^j \: < \: C_i^{d(\ell, r)} + \frac{1}{m} \ds\sum_{\tau_i^j \in \overline{\mathcal{P}_i^{d(\ell,r)}}} c_i^j $} 
\label{eq:proof3}
\end{align}
\noindent Then it follows that:
\begin{align}
& \resizebox{0.82\hsize}{!}{$\ds\sum_{\tau_i^j \in \critpp{i}} c_i^j + \frac{1}{m} \ds\sum_{\tau_i^j \in\overline{\critpp{i}}} c_i^j < \ds\sum_{\tau_i^j \in \mathcal{P}_i^{d(\ell,r)}} c_i^j + \frac{1}{m} \ds\sum_{\tau_i^j \in\overline{\mathcal{P}_i^{d(\ell,r)} } } c_i^j$} 
\label{eq:proof4}
\end{align}
\noindent Since 
\begin{align}
& \resizebox{0.88\hsize}{!}{$\ds\sum_{\tau_i^j \in \overline{\critpp{i}}} c_i^j - \ds\sum_{\tau_i^j \in \overline{\mathcal{P}_i^{d(\ell,r)}}} c_i^j = \ds\sum_{\tau_i^j \in \mathcal{P}_i^{d(\ell,r)}} c_i^j - \ds\sum_{\tau_i^j \in \critpp{i}} c_i^j$} 
\label{eq:proof5}
\end{align}
\noindent By substituting Eq.~\eqref{eq:proof5} into Eq.~\eqref{eq:proof4}, Eq.~\eqref{eq:proof3} leads us to:  
\begin{align}
& \resizebox{0.88\hsize}{!}{$\ds\sum_{\tau_i^j \in \critpp{i}} c_i^j - \frac{1}{m} \ds\sum_{\tau_i^j \in \critpp{i}} c_i^j < \ds\sum_{\tau_i^j \in \mathcal{P}_i^{d(\ell,r)}} c_i^j - \frac{1}{m} \ds\sum_{\tau_i^j \in \mathcal{P}_i^{d(\ell,r)}} c_i^j$}
\end{align}   
\noindent which trivially means $\Ccritpp{i} < C_i^{d(\ell, r)}$, contradicting the initial assumption. The Lemma follows. \end{proof}
Informally speaking Lemma~\ref{lemma:self-interference1} infers, for any non-parallel pair of fringe nodes $\tau_i^\ell$ and $\tau_i^r$, that an upper-bound on the response time of $\tau_i$ is obtained by considering $\critpp{i}$ between any $\tau_i^\ell$ and $\tau_i^r$. At the same time, the nodes which do not belong to $\critpp{i}$ are assumed to induce the maximum interference over it. As this is proven for any pair of nodes, the result also holds for the extreme nodes.
\begin{figure}[h!]
  \begin{center}
  \includegraphics[width=0.67\linewidth]{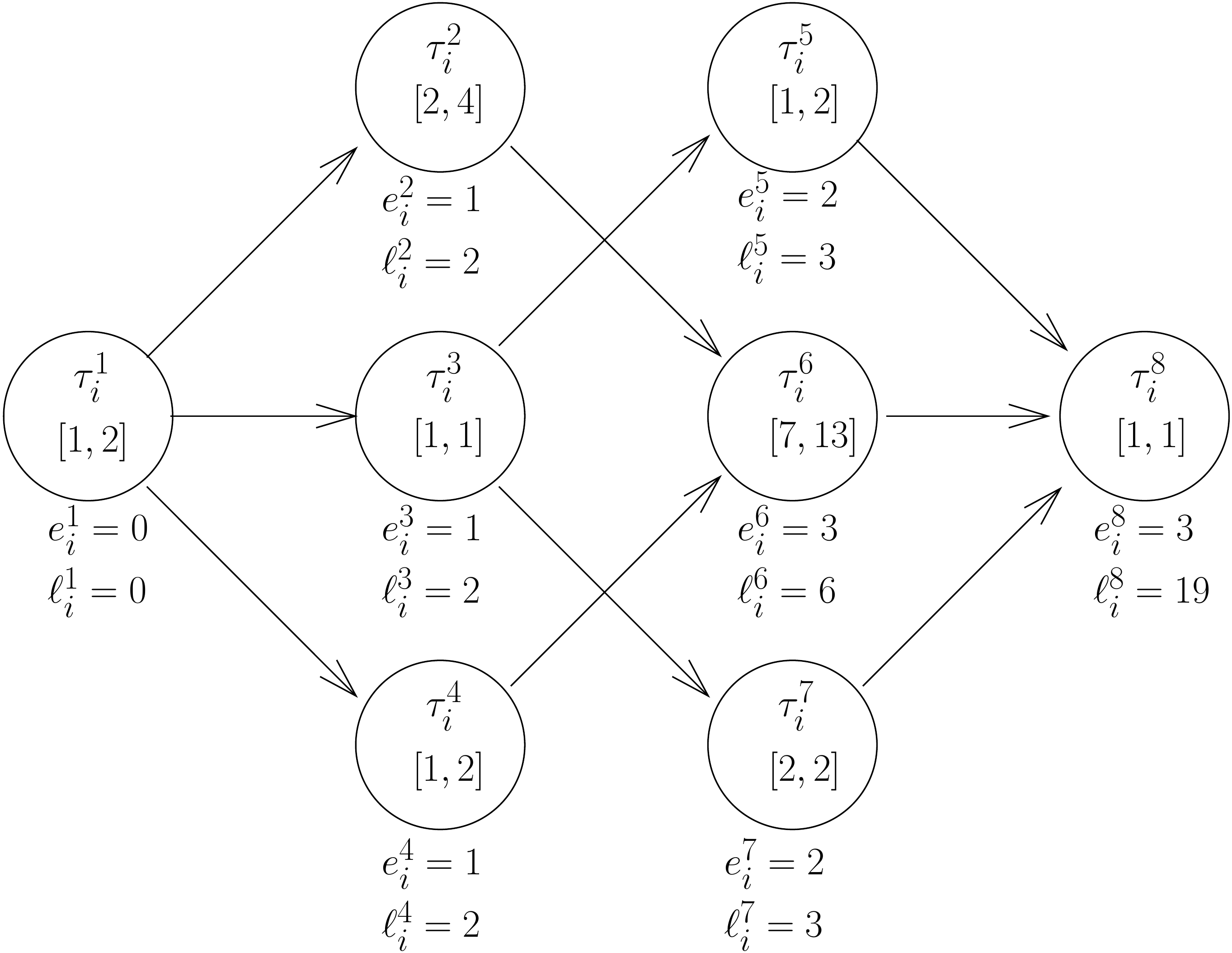}
  \caption{Earliest and latest release times for nodes in a $\da$}
  \label{fig:dagexample_values}
  \end{center}
\end{figure}
Now we focus on deriving the critical path $\critp{i}$ in $G_i$. For every node $\tau_i^j$ in $G_i$, we denote by $e_i^j$ and $\ell_i^j$ its {\em earliest} and {\em latest} release times, respectively. Note: These quantities can be computed through a breadth-first~\cite{DATE} traversal of $G_i$. Assuming $\tau_i^1$ and $\tau_i^{\mathrm{last}}$ are the entry and exit nodes of $\tau_i$, the earliest release time of any node $\tau_i^j$ without any interference can be computed as follows.
\begin{align}
&\resizebox{0.11\hsize}{!}{$ e_i^{1} \equals 0 $} \label{eq:7}\\
&\resizebox{0.48\hsize}{!}{$ e_i^j \equals \ds\max_{\tau_i^x \in \mathit{pred}(\tau_i^j)} \{ e_i^x+c_{i,\operatorname{min}}^x\} $} \label{eq:8}
\end{align}
\noindent where $c_{i,\operatorname{min}}^x$ is the minimum execution requirement of $\tau_i^x$. In the same manner, a breath-first traversal of $G_i$ starting from $\tau_i^{\mathrm{last}}$ provides the latest release time of $\tau_i^j$ as follows.
\begin{align}
&\resizebox{0.15\hsize}{!}{$ \ell_i^{\mathrm{last}'} \equals 0 $} \label{eq:9}\\
&\resizebox{0.46\hsize}{!}{$ \ell_i^{j'} \equals \ds\min_{\tau_i^x \in \mathit{succ}(\tau_i^j)} \{ \ell_i^{x'}\} -c_{i,\operatorname{min}}^j $} \label{eq:10} \\
&\resizebox{0.22\hsize}{!}{$ \ell_i^{j} = \ell_i^{j'} - \ell_i^{1'} $} \label{eq:11}
\end{align}
\noindent Eq.~\eqref{eq:8} and Eq.~\eqref{eq:11} clearly represent a {\em lower}- and an {\em upper}-bound on the best-case and worst-case start times of node $\tau_i^j$, respectively. This can be observed in the following two scenarios: $(i)$~Node $\tau_i^j$ does not suffer any external interference and all its parents request for their minimum execution requirements purveys $e_i^j$; $(ii)$~Node $\tau_i^j$ suffers the maximum possible external interference and its parents request for their maximum execution requirements purveys $\ell_i^j$. With these equations, we can derive the worst-case response time of $\tau_i$ in isolation, denoted by $R_i^{\operatorname{isol}}$. To do so, without explicitly referring to $\critp{i}$, we compute the critical path length $\maxcomp{i}$ of $\tau_i$ as the two problems can be addressed separately. From Eq.~\eqref{eq:9} and \eqref{eq:10} and by starting from the exit node of $\tau_i$, $\maxcomp{i}$ is obtained as follows.
\begin{align}
&\resizebox{0.17\hsize}{!}{$ \mathrm{exe}_i^{\mathrm{last}} \equals 0 $} \label{eq:12}\\
&\resizebox{0.62\hsize}{!}{$ \mathrm{exe}_i^{j} \equals \ds\max_{\tau_i^x \in \mathit{succ}(\tau_i^j)} \{ \ell_i^{x}\} - \mathrm{exe}_{\operatorname{i,\max}}^{j} \hspace{2.5ex} \forall \tau_i^j \in G_i $} \label{eq:13}\\
&\resizebox{0.19\hsize}{!}{$ \maxcomp{i} \equals \mathrm{exe}_i^{1} $}
\end{align}
\noindent For any $\tau_i^j \in \critp{i}$, the execution requirement of the nodes in $conc(\tau_i^j)$ is yielded by $\mathrm{SI}_i \equals C_i-\maxcomp{i}$. From Lemma~\ref{lemma:self-interference1}, an upper-bound on the response time of task $\tau_i$, including only the self-interference is thus given by:
\begin{equation}
\resizebox{0.4\hsize}{!}{$ R_i^{\operatorname{isol}}=\maxcomp{i}+ \frac{1}{m} \cdot \mathrm{SI}_i $}
 \label{eq:isol}
\end{equation}

\section{Upper-bound on the Interference and Schedulability Condition\label{sec:upper_sched}}

In this section we provide an upper-bound on the interference of any $\da$ task $	\tau_i$ and we derive a sufficient schedulability condition. To this end, we distinguish between two scenarios: $(i)$~The scenario where $\tau_i$ does not suffer any interference from higher priority tasks, and $(ii)$~The scenario where $\tau_i$ suffers the maximum possible interference. For brevity sake we assume that all tasks have carry-in at this stage, and will relax this assumption in Section~\ref{sect:reduction}.

Regarding Scenario $(i)$, we recall that $e_i^j$ is a lower-bound on the release time of node $\tau_i^j$. This leads to an upper-bound function $f^{U}_{i,j}(t)$ on the workload request of $\tau_i^j$ at any time $t$ (see~Figure~\ref{fig:Smallcase}) defined as follows.
\begin{align}
&\resizebox{0.7\hsize}{!}{$ f^{U}_{i,j}(t) \equals \min\left(\max( (t  \mod T_i) - e_i^j , 0), c_{i,\max}^j\right) $}
\label{fub}
\end{align}
\noindent Since the workload request of the $\da$ task $\tau_i$ is the sum over the workload requests of all its nodes, then an upper-bound on the workload request of $\tau_i$ at any time $t$ (see~Figure~\ref{fig:FUL}) is defined as follows.
\begin{equation} \hspace{-.03cm}
\resizebox{0.15\hsize}{!}{$ F_i^U(t) \equals $}
\begin{cases}
\resizebox{0.2\hsize}{!}{$ \sum_{\tau_i^j \in G_i} f^{U}_{i,j}(t)$} & \resizebox{0.25\hsize}{!}{$\text{ if } t < T_i - \maxcritk{i} $} \\ \\
\resizebox{0.38\hsize}{!}{$ \left\lfloor \frac{t + \maxcritk{i}}{T_i}   \right\rfloor \cdot C_i + \sum_{\tau_i^j \in G_i}  U_i^j(t) $} & \resizebox{0.15\hsize}{!}{$ \text{otherwise} $}
\end{cases}
\label{FUB}
\end{equation}

\noindent where \resizebox{0.52\hsize}{!}{$ U_i^j(t) = f^{U}_{i,j}\left((t + \maxcritk{i}) \mod T_i  \right)$} and \resizebox{0.28\hsize}{!}{$ \maxcritk{i} \equals R_i - \maxcomp{i}$}.

Regarding scenario~$(ii)$, $\tau_i^1$ is assigned to a core at most $\maxcritk{i}$ time units after the task is released  and $\tau_i^j$ is released at most $\ell_i^j$ time units after node $\tau_i^1$ has started execution. This leads to a lower-bound function $f^{L}_{i,j}(t)$ on the workload request of $\tau_i^j$ at any time $t$ (see~Figure~\ref{fig:Smallcase}) defined as follows.
\begin{align}
&\resizebox{0.7\hsize}{!}{$ f^{L}_{i,j}(t) \equals  \min\left( \max( (t  \mod T_i) - \ell_i^j ,  0), c_{i,\max}^j \right)$}
\label{flb}
\end{align}

\begin{figure}[tb]
  \begin{center}
  \includegraphics[width=0.77\linewidth]{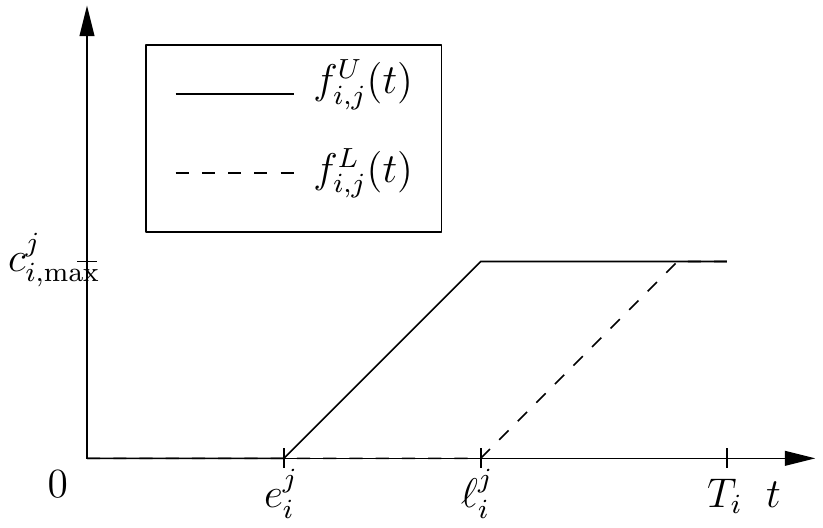}
  \caption{Extreme cases for node $\tau_i^j$ execution requirements}
  \label{fig:Smallcase}
  \end{center}
\end{figure}

\begin{figure}[tb]
  \begin{center}
  \includegraphics[width=0.95\linewidth]{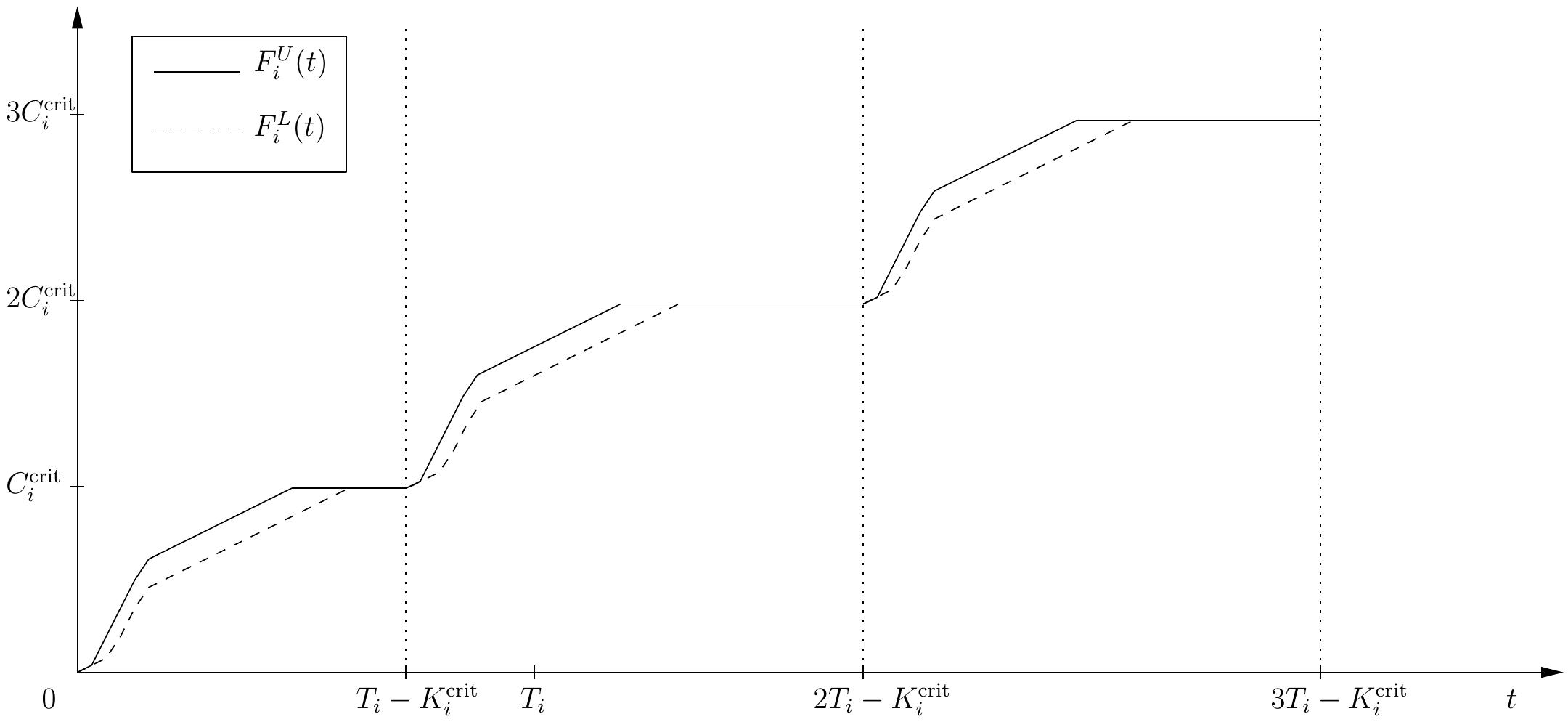}
  \caption{Extreme cases for $\da$ task $\tau_i$ execution requirements}
  \label{fig:FUL}
  \end{center}
\end{figure}

A lower-bound on the workload request of $\tau_i$ at any time $t$ (see~Figure~\ref{fig:FUL}) is thus defined as follows. 
\begin{equation} \hspace{-.03cm}
\resizebox{0.15\hsize}{!}{$ F_i^L(t) \equals $}
\begin{cases}
\resizebox{0.2\hsize}{!}{$ \sum_{\tau_i^j \in G_i} f^{L}_{i,j}(t)$} & \resizebox{0.25\hsize}{!}{$\text{ if } t < T_i - \maxcritk{i} $} \\ \\
\resizebox{0.38\hsize}{!}{$ \left\lfloor \frac{t + \maxcritk{i}}{T_i}   \right\rfloor \cdot C_i + \sum_{\tau_i^j \in G_i}  L_i^j(t) $} & \resizebox{0.15\hsize}{!}{$ \text{otherwise} $}
\end{cases}
\label{FLB}
\end{equation}

\noindent where 
\resizebox{0.52\hsize}{!}{$ L_i^j(t) = f^{L}_{i,j}\left( (t + \maxcritk{i}) \mod T_i   \right) $}.

Eq.~\eqref{FUB} and Eq.~\eqref{FLB} can be used to obtain an upper-bound on the workload request of $\tau_i$ in a time window of length~$\Delta$. To this end, we consider that an activation of $\tau_i$ occurs $\phi$ time units prior to the beginning of the targeted window. Then two situations can lead to increasing the workload request of $\tau_i$ in the window: $1)$~At the beginning of the window, say at time~$0$, $\tau_i$ suffers the maximum possible interference and its nodes are released as late as possible; $2)$~At the end of the window, say at time $\Delta$, $\tau_i$ does not suffer any interference and its nodes are released as early as possible.

\begin{lemma}[Upper-bound on the Workload of $\tau_i$ with Carry-in]
\label{lemma:2}
Assuming task $\tau_i$ has carry-in, an upper-bound on its workload request in a window of length $\Delta$ is given by:
 \begin{equation}
\resizebox{0.7\hsize}{!}{$ W_i^{\operatorname{CI}}(\Delta) \equals \ds\max_{\phi \in [0 , \maxcomp{i}]} \left\{F^{U}_i(\Delta + \phi ) - F^{L}_i(\phi ) \right\}$}
 \label{max_workload}
\end{equation}
\end{lemma}
\begin{proofsketch}
We consider an activation of $\tau_i$ occurring at time $t_r = - \maxcritk{i} = -R_i + \maxcomp{i}$. The worst-case scenario for task $\tau_i$ is when it is prevented from execution on any core by higher priority tasks in the interval $[-\maxcritk{i},0]$. Let us assume this worst-case scenario and let us assume that all nodes $\tau_i^j \in G_i$ are released at time $\ell_i^{j}$ but one specific node $\tau_i^k$ is released at time $\ell_i^{k'} < \ell_i^{k}$. Since $f^{L}_{i,k}(t)$ is a lower-bound on the workload request of $\tau_i^k$ at any time $t$, it follows that the workload executed after $t$, when $\tau_i^k$ is released at time $\ell_i^{k}$, is greater than or equal to the workload request of $\tau_i^k$ when it is released at time $\ell_i^{k'}$. Hence on the left border of the window of length $\Delta$ (i.e., at the beginning of the window), if the nodes are assumed to be released as late as possible, then the workload request in the window is maximized. On the right border of the window (i.e., at 
the end of the window), we assume the earliest release time of all the nodes $\tau_i^j \in G_i$ but one specific node $\tau_i^k$. By applying the same logic, it follows that the workload request in the window is maximized since the nodes are assumed to be released as early as possible and $f^{U}_{i,k}(t)$ is an upper-bound on the workload request of $\tau_i^k$ at any time $t$. 
		
Now, let $n_i^{\text{pmax}}$ denote the maximum number of parallel nodes in $G_i$. We recall that the summation of the workload requests of all the nodes $\tau_i^j \in G_i$ is a piecewise linear function, where each segment has its first derivative in the interval $[0, n_i^{\text{pmax}}]$. In order to compute the maximum workload request of each $\da$ task $\tau_i$ in an interval of length $\Delta$, we must evaluate the workload request in all windows of length $\Delta$ assuming an offset $\phi \geqslant 0$. Since on the one hand the first derivative of $F_i^U(.)$ (resp. the first derivative of $F_i^L(.)$) is clearly periodic from time $T_i - \maxcritk{i}$ with a period $T_i$ (see Fig.~\ref{fig:FUL}), and on the other hand, the next activation of $\tau_i$ occurs only at time $t_r^{\text{next}} = t_r + T_i = T_i - \maxcritk{i}$, it is not necessary to check the offsets $\phi$ over $\maxcomp{i}$ as there is no extra workload after $\maxcomp{i}$ by construction. Hence $\phi \in [0, \maxcomp{i}]$ and 
the lemma follows.
\end{proofsketch}

In order to obtain the solution of Eq.~\eqref{max_workload}, instead of exaustively testing all the values of $\phi$ in the continuous interval $[0, \maxcomp{i}]$, we derive the finite set $V_i(\Delta)$ of offsets $\phi$ which maximizes it hereafter.

As previously mentioned, both $F^{U}_i(\cdot)$ and $F^{L}_i(\cdot)$ are piecewise linear functions. Hence, the set of points where the first derivative of $F^{U}_i(\cdot)$ increases and the set points where the first derivative of $F^{L}_i(\cdot)$ decreases should be considered respectively at the left and at the right border of the targeted window of length $\Delta$. The points in these sets maximize the workload request in the window. Formally, let $\Gamma_i(\Delta) \equals \Gamma_i^1(\Delta) \cup \Gamma_i^2(\Delta)$ where $\Gamma_i^1(\Delta) \equals \{ \phi \in [0, \maxcomp{i}], \: \text{the first derivative of} \: F^{L}_i(\phi) \: \text{increases} \}$ and $\Gamma_i^2(\Delta) \equals \{ \phi \in [0, \maxcomp{i}], \: \text{the first derivative of} \: F^{U}_i(\Delta + \phi) \: \text{decreases}\}$. Since $0 \leqslant \phi \leqslant \maxcomp{i}$, then for each node $\tau_i^j$, the first derivative of $F^{L}_i(\cdot)$ can increase only at points $\ell_i^j$. Similarly, the first derivative of $F^{U}_i(\Delta + \
cdot)$ can decrease only at points $k T_i - \maxcritk{i} + e_i^j + c_i^j - \Delta$ such that $k \in \mathbb{N}$ and $\Delta \leqslant  k T_i - \maxcritk{i} $ $\leqslant\Delta + \maxcomp{i}$. Therefore $V_i(\Delta)$ can be defined as follows.
 \begin{align}
 &\resizebox{0.85\hsize}{!}{$V_i(\Delta) \equals \ds\bigcup_{\tau_i^j \in G_i} 
\left( \left\{\ell_i^j \right\}  \cup \left\{ k T_i - \maxcritk{i} + e_i^j + c_i^j - \Delta,  \right. \right. $} \nonumber \\
& \resizebox{0.85\hsize}{!}{$\left. \left. \:  \text{such that} \: k \in \mathbb{N} \: \text{and} \: \Delta \leqslant  k T_i - \maxcritk{i} \leqslant \Delta + \maxcomp{i}\right\} \right)$}
\end{align}
The computation of $W_i^{\operatorname{CI}}(\Delta)$ for each $\tau_i$ (with $i \in [1, n]$) makes it easy to assess an upper-bound on the interference it will induce on the workload of the lower priority tasks in any given time window. From~\cite{baker}, it has been proven that every unit of execution of a LL task can interfere for at most $\frac{1}{m}$ units on the workload request of any other LL task with a lower priority. Thus, an upper-bound on the interference suffered by the $\da$ task $\tau_i$ in a window of size $\Delta$ is provided as:
\begin{equation}
\resizebox{0.46\hsize}{!}{$ I_i(\Delta) \equals \frac{1}{m} \cdot \displaystyle\sum_{j \in hp(\tau_i)} W_j^{\operatorname{CI}}(\Delta) $}
 \label{interf}
\end{equation}
where $hp(\tau_i)$ is the set of tasks with a higher priority than $\tau_i$. A sufficient schedulability condition for a $\da$ task-set $\mathcal{T}$ is derived from Eq.~\eqref{interf} as follows.

\begin{theorem}[Sufficient schedulability condition\label{responsetime}]
A $\da$ task-set $\mathcal{T}$ is schedulable on a $m$-homogeneous multicores using a $\gftp$ scheduler if:
\begin{equation}\label{eq:sched_cond}
\resizebox{0.28\hsize}{!}{$ \forall \tau_i \in \mathcal{T}, R_i \leqslant D_i $}
\end{equation}
where $R_i$ is computed by the following fixed-point algorithm. 
\begin{equation}
 \begin{cases}
 R_i^{\{0\}} =  R_i^{\operatorname{isol}} &\text{if} \; k=0 \nonumber\\ 
 R_i^{\{k\}} = I_i\left(R_i^{\{k-1\}} \right) + R_i^{\operatorname{isol}} &\text{if} \; k \geqslant 1 \nonumber
\end{cases}
\end{equation}
Note: This iterative algorithm stops as soon as for any $k \ge 1$, $R_i^{\{k\}} = R_i^{\{k-1\}}$ or $R_i^{\{k\}} > D_i$. In the latter case, $\tau_i$ is deemed not schedulable.
\end{theorem}
\begin{proof}
This theorem follows directly from Lemma~\ref{lemma:1}, Lemma~\ref{lemma:2}, Eq.~\eqref{eq:isol} and Eq.~\eqref{interf}.
\end{proof}

\section{Reduction of the number of tasks with carry-in\label{sect:reduction}}

\begin{figure}[h!]
\begin{center}
\includegraphics[width=1\linewidth]{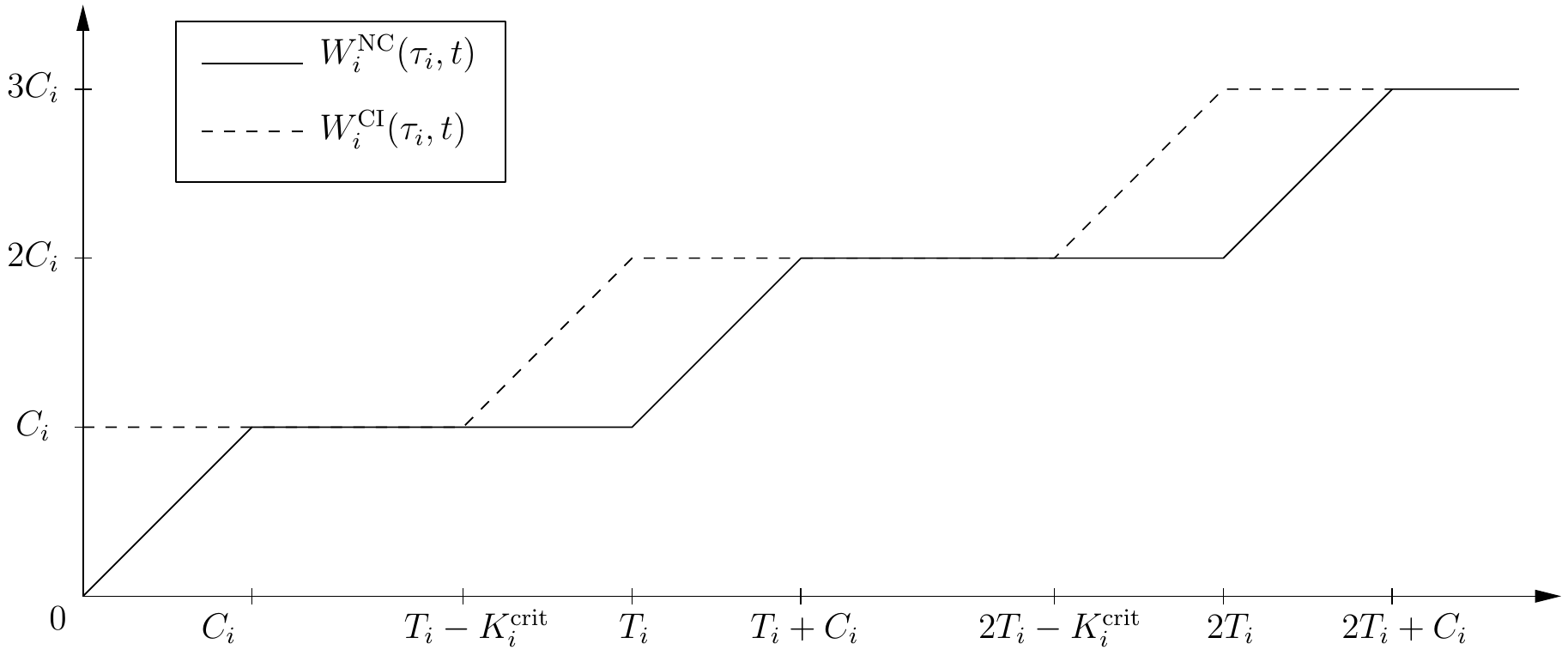}
\caption{Functions \small $W^{\operatorname{NCseq}}(\tau_j,\Delta)$ and $W^{\operatorname{CIseq}}(\tau_j,\Delta)$ for task $\tau_j$}
\label{fig:Wfunctions}
\end{center}
\end{figure}

Rather than considering that each $\da$ task has a carry-in as in Section~\ref{sec:upper_sched}, the intuitive idea of this section consists of reducing the number of tasks with carry-in to at most ($m-1$) tasks (where $m$ is the number or cores). Since it is usually the case that $m \ll n$, we thus obtain a tighter upper-bound on the interference that each task may suffer at run-time and finally a better schedulability condition for each task. To accomplish this, first let us recall some fondamental results regarding the ``Liu \& Layland (LL) task model''. 

\noindent {\bf Upper-bound on the workload request of a LL task without carry-in.} Let $\tau_j$ be a LL task with no pending workload at the beginning of a window of length $\Delta$. An upper-bound on its workload request in this window is recalled (see~\cite{BertognaRTA,Guan}):
\begin{equation}
\resizebox{0.75\hsize}{!}{$W^{\operatorname{NC-LL}}(\tau_j,\Delta) = \left\lfloor \frac{\Delta}{T_j} \right\rfloor \cdot C_j + \min(\Delta \mod T_j, C_j)$}
\label{noncarryin}
\end{equation}
\noindent {\bf Upper-bound on the workload request of a LL task with carry-in.} Let $\tau_k$ be a LL task with some pending workload at the beginning of a window of length $\Delta$. An upper-bound on its workload request in this window is recalled (see~\cite{BertognaRTA,Guan}):
\begin{align}
&\resizebox{0.65\hsize}{!}{$ W^{\operatorname{CI-LL}}(\tau_k,\Delta) = \left\lfloor \frac{\max(\Delta-C_k,0)}{T_k} \right\rfloor \cdot C_k + C_k$} \nonumber \\ 
&\resizebox{0.75\hsize}{!}{$ + \max\left((\max(\Delta - C_k, 0) \mod T_k) - (T_k - R_k), C_k\right)$}
 \label{carryin}
\end{align}
\noindent {\bf Extra workload request of a LL task.} The difference between the upper-bounds --{\em with} and {\em without}-- carry-in of a LL Task $\tau_i$ in a window of length $\Delta$ is thus recalled as:
\begin{equation}
\resizebox{0.72\hsize}{!}{$ W^{\operatorname{diff-LL}}(\tau_i,\Delta) \equals W^{\operatorname{CI-LL}}(\tau_i,\Delta) - W^{\operatorname{NC-LL}}(\tau_i,\Delta)$}
\end{equation}
\noindent {\bf Upper-bound on the interference of a LL task.} Assume a $\gftp$ scheduler and a $\da$ task-set $\mathcal{T}$ in which tasks are in a decreasing priority order. An upper-bound on the interference that higher priority tasks induce on the execution of task $\tau_i$ in a targeted window of length $\Delta$ is recalled~(see~\cite{BertognaRTA,Guan}):
\begin{align}
&\resizebox{0.75\hsize}{!}{$ I_i^{\operatorname{LL}}(\Delta) \equals  \frac{1}{m} \cdot \left( \ds\sum_{l=1}^{m-1} \max^l_{ \tau_j \in \{\tau_1,\ldots, \tau_{i-1}\}} W^{\operatorname{diff-LL}}(\tau_j,\Delta) \right) $} \nonumber \\ 
&\phantom{xxxxxx} \resizebox{0.48\hsize}{!}{$ + \ds\sum_{\tau_j \in \{\tau_1,\ldots, \tau_{i-1} \}} W^{\operatorname{NC-LL}}(\tau_j,\Delta)$}
\label{SEQ_I}
\end{align}
In Eq.~\eqref{SEQ_I}, $\ds\max^l_{\tau_h \in \{\tau_1,\ldots, \tau_{i-1}\}}(\cdot)$ returns the $\ell^{th}$ greatest value among the workload of tasks with a higher priority than $\tau_i$. For a LL task-set, it has been proven in~\cite{Guan} that a worst-case scenario in terms of total workload request in a targeted window of length $\Delta$ can be constructed by considering ($m-1$) tasks with carry-in. Therefore, it follows that the workload induced by these carry-in tasks in this window of concern cannot exceed the difference between ($i$)~the maximum workload assuming no carry-in for all tasks (see~Eq.~\eqref{noncarryin}) and ($ii$)~the workload assuming the carry-in scenario (see~Eq.~\eqref{carryin}). Consequently, from the view-point of task $\tau_i$, if $i < m$, then $\tau_i$ does not suffer any interference, otherwise, if $i \geqslant m$, then we can choose the ($m-1$) tasks among $\{\tau_1,\ldots, \tau_{i-1}\}$ such that the difference between the workload assuming the non-carry-in scenario and the 
workload assuming the carry-in scenario is the largest possible for each selected task. By summing up these differences and the remaining ``$(i-1)-(m-1)= i-m$'' workloads corresponding to the tasks without carry-in, an upper-bound on the workload that higher priority tasks induce in the window of length $\Delta$ is computed.

Before we extend Eq.~\eqref{SEQ_I} to the scheduling problem of $\da$ tasks using a $\gftp$ scheduler, let us present an alternative formal proof to the one provided by Guan et al.~\cite{Guan} for the analysis considering ($m-1$) tasks with carry-in.

\begin{theorem}[Eq.~\eqref{SEQ_I} is an Upper-bound for LL tasks~\cite{Guan}\label{m-1interference}]
Let $\tau$ be a feasible LL task-set scheduled by using a $\gftp$ scheduler on a $m$-homogeneous multicores. Let task $\tau_i \in \tau$. Eq.~\eqref{SEQ_I} is an upper-bound on the interference on $\tau_i$ in any window of length $\Delta$.
\end{theorem}
\begin{proof}
Since $\tau$ is feasible, let $t_0$ be the latest time-instant such that at least one core is idle at time $t_0 - \epsilon, \forall \epsilon \geqslant 0$, then at most ($m-1$) tasks have a carry-in workload at time instant $t_0 - \epsilon, \forall \epsilon \geqslant 0$. Let $\Delta_0$ be the window of length $\Delta$ starting at $t_0$. By considering the ($m-1$) tasks with the largest possible carry-in, we are conservative w.r.t. the workload request of the tasks with carry-in in $\Delta_0$. In the same vein, by considering ($i-m$) tasks without carry-in to be simultaneously released at time $t_0$ with the future activations of each of these tasks to occur as soon as it is legally permitted to do so, we are also conservative w.r.t. the workload request of the tasks without carry-in in $\Delta_0$.

Now let $\Delta_1$ be a window of length $\Delta$ starting at time $t_1 \geqslant t_0$ with the offset $\Phi \equals t_1 - t_0$. Assume that the beginning of $\Delta_1$ triggers the first activation of $\tau_i$. The earliest time-instant at which $\tau_i$ may start executing is at $t_f$ such that $t_f \geqslant \max(t_1, t_0 + \Delta)$. Indeed: ($i$)~$t_f \geqslant t_1$ (as $\tau_i$ cannot start executing before its activation time), and ($ii$)~$t_f \geqslant t_0 + \Delta$ (as all the $m$ cores are busy executing higher priority tasks between $t_0$ and $t_0 + \Delta$), by construction. As all the $m$ cores are busy executing higher priority tasks between $t_0$ and $t_1$, getting the first activation of $\tau_i$ at any time-instant in the interval $[t_0, t_f]$ (i.e., by sliding $\Delta_1$ towards $\Delta_0$), we can only increase the interference on $\tau_i$ (as the end of the execution of $\tau_i$ remains unchanged). The maximum interference is obtained when $\tau_i$ is activated simultaneously with all 
higher priority tasks, i.e., at time $t_0$ as then we have the largest possible carry-in as well as non-carry-in interference on the execution of $\tau_i$. The theorem follows.  
\end{proof}

\section{Extension to $\da$-based Tasks}
In this section we extend the reduction of the number of tasks with carry-in obtained in the framework of LL tasks to the $\da$ task model. To accomplish this end, we distinguish between the upper-bound on the workload request of the tasks {\em with} carry-in (see Eq.~\eqref{max_workload}) and {\em without} carry-in (which is detailed hereafter). These expressions will be considered when computing the interference of higher priority tasks on the execution of every $\da$ task $\tau_i$ in a window of length $\Delta$.

\noindent {\bf Upper-bound on the workload request of $\da$ tasks {\em without} carry-in.} Let us assume that $\tau_i$ is a $\da$ task {\em without} carry-in. An upper-bound on its workload request in a targeted window of length $\Delta$ can be constructed by distinguishing between the same two scenarios as those which allowed us to derive Eq.~\eqref{max_workload} in Section~\ref{sec:upper_sched}.

Regarding Scenario~($i$) where $\tau_i$ does not suffer any interference from higher priority tasks, an upper-bound on the workload request of $\tau_i$ at any time $t$ is defined as follows. 
\begin{equation} 
\resizebox{0.75\hsize}{!}{$ F_i^{\operatorname{U-NC}}(t) \equals \left\lfloor \frac{t}{T_i} \right\rfloor \cdot C_i + \sum_{\tau_i^j \in G_i} f^{U}_{i,j}(t \mod T_i) $} 
    \label{FUNC}
\end{equation}

In a similar manner, regarding Scenario~($ii$) where $\tau_i$ suffers the maximum interference from higher priority tasks, an lower-bound on the workload request of $\tau_i$ at any time $t$ is defined as follows. 
\begin{equation} 
\resizebox{0.75\hsize}{!}{$ F_i^{\operatorname{L-NC}}(t) \equals \left\lfloor \frac{t}{T_i} \right\rfloor \cdot C_i +  \sum_{\tau_i^j \in G_i} f^{L}_{i,j}(t \mod T_i) $}
\label{FLNC}
\end{equation}

As for the carry-in tasks case, Eq.~\eqref{FUNC} and Eq.~\eqref{FLNC} can be used to obtain an upper-bound on the workload request of $\tau_i$ in a time window of length $\Delta$ as claimed in Lemma~\ref{lemma:4}.
\begin{lemma}[Upper-bound on the Workload of $\tau_i$ Without Carry-in]
\label{lemma:4}
Assiming no carry-in of task $\tau_i$, an upper-bound on its workload request in a window of length $\Delta$ is given by:
 \begin{equation}
\resizebox{0.75\hsize}{!}{$ W_i^{\operatorname{NC}}(\Delta) \equals \ds\max_{\phi \in [0, \maxcomp{i}]} \left\{F^{\operatorname{U-NC}}_i(\Delta + \phi) - F^{\operatorname{L-NC}}_i(\phi) \right\} $}
 \label{max_workloadNCDAG}
\end{equation}
\end{lemma}

\begin{proofsketch}
The proof sketch of this lemma follows the same reasoning as that of Lemma~\ref{lemma:2}.
\end{proofsketch}

From Lemma~\ref{lemma:2} and Lemma~\ref{lemma:4}, it follows that the difference between the upper-bounds --{\em with} and {\em without}-- carry-in for a $\da$ task $\tau_i$ in a window of length $\Delta$ is can be written as:
\begin{equation}
\resizebox{0.75\hsize}{!}{$  W^{\operatorname{diff-DAG}}(\tau_i,\Delta) \equals W^{\operatorname{CI}}(\tau_i,\Delta) - W^{\operatorname{NC}}(\tau_i,\Delta) $}
 \label{DIFFDAG}
\end{equation}

All the results presented so far enable us to present a tighter upper-bound on the interference of a $\da$ task $\tau_i$ together with the corresponding sufficient schedulability condition. 

\noindent {\bf Tighter Upper-bound on the Interference of a $\da$ Task.} Assume a $\gftp$ scheduler and a $\da$ task-set $\mathcal{T}$ in which tasks are in a decreasing priority order as in Section~\ref{sect:reduction}. An upper-bound on the interference that higher priority tasks induce on the execution of task $\tau_i$ in a targeted window of length $\Delta$ is obtained as follows.
\begin{align}
&\resizebox{0.8\hsize}{!}{$ I_i^{\da}(\Delta) \equals \frac{1}{m} \cdot \left(\ds\sum_{l=1}^{m-1} \max^l_{\tau_j \in \{\tau_1,\ldots, \tau_{i-1}\}}  W^{\operatorname{diff-DAG}}(\tau_j,\Delta) \right) $} \nonumber \\ 
&\phantom{xxxxxx} \resizebox{0.48\hsize}{!}{$ + \ds\sum_{\tau_j \in \{\tau_1, \ldots, \tau_{i-1}\}} W^{\operatorname{NC}}(\tau_j,\Delta)$}
\label{DAG_I}
\end{align}

Each term in Eq.~\eqref{DAG_I} is explained as the corresponding term in Eq.~\eqref{SEQ_I} and a tighter schedulability test follows.

\begin{theorem}[Tighter Sufficient Schedulability Condition\label{responsetime2}]
A $\da$ task-set $\mathcal{T}$ is schedulable on a $m$-homogeneous multicores using a $\gftp$ scheduler if:
\begin{equation}\label{eq:sched_cond2}
\resizebox{0.28\hsize}{!}{$ \forall \tau_i \in \mathcal{T}, R_i \leqslant D_i $}
\end{equation}
where $R_i$ is computed by the following fixed-point algorithm. 
\begin{equation}
 \begin{cases}
 R_i^{\{0\}} = R_i^{\operatorname{isol}} &\text{if} \; k=0 \nonumber\\ 
 R_i^{\{k\}} = I_i^{\da} \left(R_i^{\{k-1\}} \right) + R_i^{\operatorname{isol}} &\text{if} \; k \geqslant 1 \nonumber
\end{cases}
\end{equation}
Note: This algorithm also stops as soon as for any $k \ge 1$, $R_i^{\{k\}} = R_i^{\{k-1\}}$ or $R_i^{\{k\}} > D_i$. Again, in the latter case, $\tau_i$ is deemed not schedulable.
\end{theorem}

\begin{proof}
The proof of this theorem is similar to that of Theorem~\ref{m-1interference}. The difference here resides in the evaluation of the upper-bound on the workload of tasks without carry-in. Instead of considering a synchronous activation at these tasks at the beginning of the targeted window and assume their subsequent activations to occur as soon as it is legally permitted to do so, the upper-bound has to be computed by using Eq.~\ref{max_workloadNCDAG}).
\end{proof}

\section{conclusions}

In this paper, a sufficient schedulability test for fully preemptive $\da$-based tasks with constrained deadlines is presented. A global fixed task priority ($\gftp$) scheduler and a homogeneous multicore platform are assumed. Under these settings, this work is the first to address this problem to the best of our knowledge. As future work we intend to evaluate the properties of a task model where nodes belonging to each task may execute with different priorities rather than directly inheriting their priority from the task they belong to. 

\bibliographystyle{plain}

\bibliography{defs-abbrev,jmbib}

\end{document}